\documentclass{birkjour}

\usepackage{epsfig}
\usepackage{graphicx}
\usepackage{amssymb}
\usepackage{amsmath}
\usepackage{amsfonts}
\usepackage{mathrsfs}
\usepackage[dvips]{color}
\usepackage{cite}

\newcommand{\R}{\mathbb{R}}
\newcommand{\C}{\mathbb{C}}

\newcommand{\fn}{{\,\mathfrak{n}}}

\newcommand{\fz}{\mathfrak{z}}

\newcommand{\bn}{{\mathbf{n}}}

\newcommand{\bI}{\mathbf{I}}
\newcommand{\bJ}{\mathbf{J}}

\newcommand{\bK}{\mathbf{K}}
\newcommand{\bM}{\mathbf{M}}

\newcommand{\cE}{\mathcal{E}}

\newcommand{\cN}{\mathcal{N}}

\newcommand{\cP}{\mathcal{P}}
\newcommand{\cQ}{\mathcal{Q}}
\newcommand{\cR}{\mathcal{R}}
\newcommand{\cS}{\mathcal{S}}
\newcommand{\cT}{\mathcal{T}}

\newcommand{\cX}{\mathcal{X}}
\newcommand{\be}{\begin{equation}}
\newcommand{\ee}{\end{equation}}
\newcommand{\bea}{\begin{eqnarray}}
\newcommand{\eea}{\end{eqnarray}}
\newcommand{\nn}{\nonumber}
\newcommand{\kt}{\rangle}
\newcommand{\bra}{\langle}
\newcommand{\ed}{\end{document}}

\newcommand{\bi}{\begin{itemize}}
\newcommand{\ei}{\end{itemize}}

\newcommand{\bce}{\begin{center}}
\newcommand{\ece}{\end{center}}
\newcommand{\sA}{\mathscr{A}}

\newcommand{\sE}{\mathscr{E}}

\newcommand{\sH}{\mathscr{H}}

\newcommand{\etap}{{\eta_{_+}}}

 \newtheorem{thm}{Theorem} 

 \theoremstyle{definition}
 \newtheorem{defn}{Definition}
 \theoremstyle{remark}

\begin{document}

%
%
%
%
%
%
%
%
%

\title{Physics of Spectral Singularities}

\author{Ali Mostafazadeh}

\address{Departments of Mathematics and Physics\\
Ko\c{c} University \\
Sar{\i}yer 34450, Istanbul\\
Turkey}

\email{amostafazadeh@ku.edu.tr}

\date{~}

\begin{abstract}

Spectral singularities are certain points of the continuous spectrum of generic complex scattering potentials. We review the recent developments leading to the discovery of their physical meaning,  consequences, and generalizations. In particular, we give a simple definition of spectral singularities, provide a general introduction to spectral consequences of $\cP\cT$-symmetry (clarifying some of the controversies surrounding this subject), outline the main ideas and constructions used in the pseudo-Hermitian representation of quantum mechanics, and discuss how spectral singularities entered in the physics literature as obstructions to these constructions. We then review the transfer matrix formulation of scattering theory and the application of complex scattering potentials in optics. These allow us to elucidate the physical content of spectral singularities and describe their optical realizations. Finally, we survey some of the most important results obtained in the subject, drawing special attention to the remarkable fact that the condition of the existence of linear and nonlinear optical spectral singularities yield simple mathematical derivations of some of the basic results of laser physics, namely the laser threshold condition and the linear dependence of the laser output intensity on the gain coefficient.

\end{abstract}

\subjclass{34L25, 47A40, 78A60}
\keywords{Spectral singularity, complex potential, scattering, zero-width resonance, $\cP\cT$-symmetry, pseudo-Hermitian operator, laser, antilaser}

\maketitle

\section{Introduction}

The term `spectral singularity' entered mathematical literature through the work of Jack Schwartz \cite{schwartz} who coined this name for a mathematical object discovered by Mark Aronovich Naimark in 1954 \cite{naimark}. Naimark had come across spectral singularities and worked out some of their consequences in his attempts at generalizing the well-known spectral theory of self-adjoint Schr\"odinger operators,
    \be
    H=-\frac{d^2}{dx^2}+v(x),
    \label{sch-op}
    \ee
defined on the half-line, i.e., for $x\in [0,\infty)$, to situations where $v(x)$ was a complex scattering potential. This marks the starting point of a comprehensive study of spectral singularities that has attracted the attention of mathematicians for over half a century \cite{kemp,lyantse,gasimov-1968,gasimov-maksudov,gasimov,langer,nagy,bairamov}. For further references, see \cite{guseinov}.

Following the pioneering work of Naimark, the notion of spectral singularity was generalized in various directions \cite{kemp,gasimov-1968,gasimov-maksudov,langer,nagy,bairamov}. In particular, Kemp \cite{kemp} considered spectral singularities of the Schr\"odinger operators (\ref{sch-op}) defined on the full line, i.e., $x\in\R$. These admit a simple description in terms of certain solutions of the Schr\"odinger equation \cite{guseinov}
    \be
    -\psi''(x)+v(x)\psi(x)=k^2\psi(x),\hspace{1cm} x\in\R,
    \label{sch-eq}
    \ee
called the Jost solutions.

Let $v(x)$ be a real or complex scattering potential defined on $\R$, and suppose that $|v(x)|\to 0$ as $x\to\pm\infty$ in such a manner that \cite{kemp}
    \be
    \int_{-\infty}^\infty (1+|x|)|v(x)|dx<\infty.
    \label{condi-1}
    \ee
Then for each $k\in\R$, the Schr\"odinger equation (\ref{sch-eq}) admits a pair of solutions $\psi_{k\pm}$ fulfilling the asymptotic boundary conditions \cite{kemp}:
    \begin{align}
    &\lim_{x\to\pm\infty} e^{\mp ikx}\psi_{k\pm}(x)=1, &&
    \lim_{x\to\pm\infty} e^{\mp ikx}\psi'_{k\pm}(x)=\pm i k.
    \label{jost}
    \end{align}
These are the celebrated \emph{Jost solutions} of (\ref{sch-eq}).
    \begin{defn}
    \label{defn-ss}
    Let $v:\R\to\C$ be a function satisfying (\ref{condi-1}) and $H$ be the Schr\"odinger operator (\ref{sch-op}) that is defined by $v$ on $\R$. A real and positive number $k_\star^2$ is called a \emph{spectral singularity} of $H$ or $v$, if the Jost solutions $\psi_{k_\star\pm}$ of (\ref{sch-eq}) are linearly dependent.
    \end{defn}

It is not difficult to see that the Jost solutions correspond to the scattering states of the potential $v(x)$. But as we shall see below, for real scattering potentials they are always linearly independent. This is why physicists did not pay much attention to spectral singularities throughout the twentieth century.

The situation began to change in 1998 by the discovery of a class of complex potentials that possessed a real spectrum \cite{bender-1998}. A well-known example is $v(x)=ix^3$, whose spectrum is discrete, real, and positive  \cite{dorey-2001}. This unexpected result was initially associated with the fact that the corresponding Schr\"odinger operator (\ref{sch-op}) was invariant under the parity-time-reversal transformation, also known as spacetime reflection \cite{bender-1998},
    \be
    \psi(x)\longrightarrow\cP\cT\psi(x)=\psi(-x)^*,
    \label{PT=}
    \ee
where $\psi$ is an arbitrary square-integrable function, i.e., $\psi\in L^2(\R)$, and $\cP$ and $\cT$ are respectively the parity and time-reversal operators defined by
    \begin{align}
    &\cP\psi(x):=\psi(-x)  , && \cT\psi(x):=\psi(x)^*.
    \label{P-T=}
    \end{align}

We postpone the discussion of the spectral implications of $\cP\cT$-symmetry to Section~\ref{section2}. Here we suffice to mention that during the last 16 years there has been a great interest in the study of $\cP\cT$-symmetric potentials. A substantial amount of the early work on this topic consisted of searching for other examples of complex potentials possessing a real spectrum. A rather straightforward method of constructing such potentials is to generate them from real potentials via non-unitary similarity transformations. The simplest example is a complex translation of the form $v(x)\to v(x+\fz)$, where $\fz$ is a complex parameter \cite{levi-znojil-2000}. To the best of our knowledge, it was in this context that spectral singularities entered in the physics literature.

In 2005 Boris Samsonov noted that the application of complex translations on a real potential could yield complex scattering potentials supporting spectral singularities \cite{samsonov-2005}. He also proposed means of removing these spectral singularities by performing certain supersymmetry transformations. In Samsonov's words this meant ``curing'' a ``disease'', for he maintained that ``Hamiltonians with spectral singularities are `bad'.''  This is a typical reaction when one encounters a `singularity'. However, the history of science teaches us that some of the greatest discoveries of mankind have their root at unwanted `singularities'. One of the aims of the present article is to show that the same applies to spectral singularities, as they provide the mathematical basis for one of the most important discoveries of all times, namely lasers.

\section{$\cP\cT$-Symmetry versus Pseudo-Hermiticity}
\label{section2}

Samsonov's article \cite{samsonov-2005} failed to provide the necessary incentive for the study of the physical aspects of spectral singularities. But soon after, spectral singularities were to reveal their presence in the study of a delta-function potential with an imaginary coupling constant, i.e., $v(x)=i\alpha\,\delta(x)$ with $\alpha\in\R$, \cite{jpa-2005b}. The motivation for this study was provided by attempts at finding a set of necessary and sufficient conditions for the reality of the spectrum of a non-Hermitian linear operator $H$. A well-advertised claim is that $\cP\cT$-symmetry provides such a condition \cite{bender-mannheim}. This is certainly not true if we take (\ref{P-T=}) as the definition of $\cP$ and $\cT$, for there are infinity of examples of real potentials, such as $v(x)=x^2+\sin x$, that do not commute with $\cP\cT$ but have a real spectrum.

In order to ensure the validity of the above claim, we need to reinterpret what we mean by $\cP\cT$-symmetry or generalize it appropriately. First, we recall the following obvious consequences of (\ref{P-T=}).
    \begin{align}
    & [\cP,\cT]=0, && \cP^2=\cT^2=(\cP\cT)^2=I,
    \label{basics}
    \end{align}
where $I$ stands for the identity operator. The following is a precise definition of $\cP\cT$-symmetry.
     \begin{defn}
     Let $\cP$ and $\cT$ be the linear operators defined on $L^2(\R)$ by (\ref{P-T=}). Then a linear operator $H$ acting in $L^2(\R)$ is said to be \emph{$\cP\cT$-symmetric}, if it commutes with $\cP\cT$, i.e., $[H,\cP\cT]=0$. Moreover, suppose that $H$ has a discrete and non-degenerate spectrum, and there is a complete set\footnote{This means that the span of this set is dense in $L^2(\R)$.} of eigenvectors $\psi_n$ of $H$ that are also eigenvectors of $\cP\cT$. Then $H$ is said to have an \emph{unbroken} or \emph{exact  $\cP\cT$-symmetry}.
     \end{defn}
For the cases where $H$ has an exact $\cP\cT$-symmetry, there are $\epsilon_n\in\C$ such that $\cP\cT\psi_n=\epsilon_n\psi_n$. Then, in view of (\ref{basics}), $\psi_n=(\cP\cT)^2\psi_n=\cP\cT(\epsilon_n\psi_n)=|\epsilon_n|^2\psi_n$. This shows that $\epsilon_n=e^{i\alpha_n}$ for some real number $\alpha_n$. Now, setting $\tilde\psi_n:=e^{i\alpha_n/2}\psi_n$, we find $\cP\cT\tilde\psi_n=e^{-i\alpha_n/2}\cP\cT\psi_n=
e^{i\alpha_n/2}\psi_n=\tilde\psi_n$. Therefore, exact $\cP\cT$-symmetry of $H$ means the existence of a complete set of $\cP\cT$-invariant eigenvectors of $H$. This argument relies only on the fact that $\cP\cT$ is an antilinear operator\footnote{An antilinear operator $\cX$ acting in a complex vector space $V$ is one whose domain is a subspace of $V$ and fulfills the antilinearity condition: $\cX(\alpha_1 v_1+\alpha_2 v_2)=\alpha_1^*\cX v_1+\alpha_2^*\cX v_2$ for all $\alpha_1,\alpha_2\in\C$ and $v_1,v_2\in V$.} squaring to $I$. Therefore, it applies to all such operators
$\cX$.\footnote{Because every antilinear operator acting in $L^2(\R)$ can be expressed in the form $\cQ\cT$ for some linear operator $\cQ$, a natural choice of notation for the antilinear operators $\cX$ is $\cQ\cT$ \cite{jpa-2008}.} We use this observation to introduce the notion of exact antilinear symmetry.
    \begin{defn}
    Let $H$ and $\cX$ be respectively linear and antilinear operators acting in a Hilbert space $\sH$, and $I$ be the identity operator on $\sH$. $H$ is said to be $\cX$-symmetric, if
    $[\cX,H]=0$. Furthermore, suppose that $H$ has a discrete spectrum and $\cX^2=I$. Then $H$ is said to have an exact $\cX$-symmetry if there is a complete set of eigenvectors $\psi_n$ of $H$ that are invariant under $\cX$, i.e., $\cX\psi_n=\psi_n$.
    \end{defn}
The following is a useful property of exact antilinear symmetry. Its application to $\cP\cT$-symmetry is the reason for the claim that exact $\cP\cT$-symmetry implies the reality of the spectrum.
    \begin{thm} Eigenvalues of every linear operator $H$ that has an exact antilinear symmetry are real.
    \label{thm1.4}
    \end{thm}
    \begin{proof}
    Let $E_n$ be an eigenvalue of $H$ and $\psi_n$ be a corresponding $\cX$-invariant eigenvector. Then, in view of the fact that $\cX\psi_n=\psi_n$ and $[\cX,H]=0$,
    \bea
    E_n^*&=&
    \frac{\bra\psi_n|E_n^*\cX\psi_n\kt}{\bra\psi_n|\psi_n\kt}=
    \frac{\bra\psi_n|\cX E_n\psi_n\kt}{\bra\psi_n|\psi_n\kt}=
    \frac{\bra\psi_n|\cX H\psi_n\kt}{\bra\psi_n|\psi_n\kt}\nn\\
    &=&  \frac{\bra\psi_n|H\cX\psi_n\kt}{\bra\psi_n|\psi_n\kt}=
    \frac{\bra\psi_n|H\psi_n\kt}{\bra\psi_n|\psi_n\kt}=E_n.\nn
    \eea
    \end{proof}

This theorem suggests generalizing $\cP\cT$-symmetry to the presence of an antilinear symmetry.
In order to avoid using the same symbols for different concepts, we use `$PT$-symmetry' to refer to this generalization.
    \begin{defn}
    We say that $H$ is \emph{$PT$-symmetric} if it has an exact antilinear symmetry.
    \label{PT-defn}
    \end{defn}
In view of this definition, Theorem~\ref{thm1.4} is equivalent to the statement that $PT$-symmetry is a sufficient condition for the reality of the spectrum of $H$. The converse can also be established if $\sH$ is finite-dimensional or if we impose a further technical condition\footnote{This is the condition of the existence of a set of eigenvectors of $H$ that forms a Riesz basis, i.e., it can be mapped to an orthonormal basis by an invertible bounded operator \cite{review}.} on $H$, \cite{p3}. Therefore, introducing the above notion of $PT$-symmetry secures the validity of the claim that it is a necessary and sufficient condition for the reality of the spectrum of a large class of linear operators. However the price one pays for doing so is a clear distinction between $PT$-symmetry and the parity-time-reversal (spacetime reflection) symmetry that we label by $\cP\cT$-symmetry. Indeed unlike Hermiticity which is a sufficient condition for the reality of the spectrum, $PT$-symmetry is both necessary and sufficient. But this $PT$-symmetry does not mean spacetime reflection $(\cP\cT)$ symmetry. Similarly to Hermiticity exact $\cP\cT$-symmetry is a sufficient but not necessary condition for the reality of the spectrum of a linear operator. Non-exact $\cP\cT$-symmetry, which can be immediately checked for a given operator, is neither necessary nor sufficient.

A major difficulty with the above notion of $PT$-symmetry is that there does not exist a universal choice for the antilinear symmetry appearing in Definition~\ref{PT-defn}; each $PT$-symmetric linear operator $H$ has its own set of antilinear operators $\cX$ which make $H$ exactly $\cX$-symmetric. If we denote this set by $\cS_H$, then $PT$-symmetry of $H$ is equivalent to the condition that $\cS_H$ is nonempty. In practice, in order to check if this is the case, one must determine an appropriate set of eigenvectors $\psi_n$ of $H$, make sure that they form a complete set, and try to construct an antilinear operator $\cX$ that leaves $\psi_n$'s invariant and squares to $I$. This is generally a difficult task.

The question of finding a necessary and sufficient condition for the reality of the spectrum of a non-Hermitian operator has a more illuminating answer.
    \begin{thm}
    Let $H$ be a linear operator with a complete set of eigenvectors that acts in a finite-dimensional inner-product space. Then $H$ has a real spectrum if and only if there is a
    positive-definite operator $\etap$ intertwining $H$ and its adjoint $H^\dagger$, i.e.,
        \be
    	H^\dagger\etap=\etap\!H,
    	\label{ph-2}
    	\ee
    alternatively $H^\dagger$ is related to $H$ by the similarity transformation:
    	\be
    	H^\dagger=\etap\!H\,\eta_{_+}^{-1}.
    	\label{ph-1}
    	\ee
    \label{thm1}
    \end{thm}
Conditions~(\ref{ph-2}) and (\ref{ph-1}), that we call `$\etap$-pseudo-Hermiticity' of $H$, was derived in Refs.~\cite{p1,p2,p3} for a more general class of operators. These act in a possibly infinite-dimensional Hilbert space $\sH$, have a discrete spectrum, and possess a complete biorthonormal eigensystem $\{(\psi_n,\phi_n)\}$, \cite{review}. The latter means the existence of a sequence of complex numbers $\{E_n\}$ and a pair of sequences of vectors, $\{\psi_n\}$ and $\{\phi_n\}$, which satisfy
	\begin{align}
	& H\psi_n=E_n\psi_n, && H^\dagger\phi_n=E_n^*\phi_n,
	\label{eg-va}\\
	& \bra\psi_m|\phi_n\kt=\delta_{mn}, &&
	\sum_n |\psi_n\kt\bra\phi_n|=I.
	\label{bior}
	\end{align}
Here $\bra\cdot|\cdot\kt$ stands for the inner product of $\sH$. We use the term `diagonalizable' to mean that $H$ admits a complete biorthonormal eigensystem.

Theorem~\ref{thm1} admits an infinite-dimensional generalization provided that we impose further restrictions on $H$ and $\etap$ \cite{review}. For pedagogical reasons we postpone discussing these to the final three paragraphs of this section. For the moment, we follow the physicists' tradition of assuming that what we know about finite dimensions is essentially valid in infinite dimensions. For example, we take the following condition as the definition of a Hermitian or self-adjoint operator:
    \[\bra\phi|H\psi\kt=\bra H\phi|\psi\kt,\]
where $\phi$ and $\psi$ are arbitrary elements of $\sH$ (hence neglecting domain issues).

A key observation made in \cite{p1,p2,p3} is that the reality of the spectrum of $H$ is related to the fact that we can turn it into a Hermitian operator by modifying the inner product of $\sH$ properly. Using the term `Hermitizablility' for the latter property, we can say that \emph{a diagonalizable operator with a real spectrum need not be Hermitian, but it is necessarily Hermitizable}. Conversely, \emph{every Hermitizable operator is diagonalizable and has a real spectrum.} Therefore, \emph{Hermitizability is a necessary and sufficient condition for the reality of the spectrum of $H$}.\footnote{A rigorous extension of this result to infinite dimensions requires special care. In mathematics literature, it is studied in the context of `symmetrizable' \cite{istratescu} and `quasi-Hermitian operators' \cite{dieudonne}. For more recent developments, see \cite{AK,AT1,AT2}.}

The modified inner product that achieves the Hermitization of $H$ is determined by the operator $\etap$ according to
	\be
	\bra\psi,\phi\kt_{\eta_+}:=\bra\psi|\etap\phi\kt,
	\label{mod-inn}
	\ee
where $\psi,\phi\in\sH$ are arbitrary. In other words, if $\sH_{\eta_+}$ labels the Hilbert space obtained by endowing the set of vectors belonging to $\sH$ with the inner product $\bra\cdot,\cdot\kt_{\eta_+}$, then $H:\sH_{\eta_+}\to\sH_{\eta_+}$ is Hermitian.

The operator $\etap$ that defines the modified inner product (\ref{mod-inn}) is usually called a `metric operator.' It is not difficult to show that its positive square root, $\rho:=\sqrt\etap$, defines a unitary operator mapping $\sH_{\eta_+}$ onto $\sH$ and that $h:=\rho\,H\rho^{-1}$ is a Hermitian operator acting in $\sH$, \cite{jpa-2003}. This provides a direct evidence for the reality of the spectrum of $H$, for $H$ and $h$ are isospectral. It also makes a connection with an earlier work on quasi-Hermitian operators that is done in the context of nuclear physics \cite{SGH}.

Another useful result of Refs.~\cite{p1,p2,p3} is the following spectral representation of the metric operator.
	\be
	\etap=\sum_n|\phi_n\kt\bra\phi_n|.
	\label{spec-rep}
	\ee
Because the eigenvectors $\phi_n$ of $H^\dagger$ are not unique, this equation signifies the non-uniqueness of the metric operator \cite{jmp-2003,jpa-2003}. Different choices for $\{\phi_n\}$ determine different metric operators for $H$. Once such a choice is made, we can construct the corresponding Hilbert space $\sH_{\eta_+}$ and view $H$ as a Hamiltonian operator acting in $\sH_{\eta_+}$. By construction $H:\sH_{\eta_+}\to \sH_{\eta_+}$ is a Hermitian operator. Therefore $(\sH_{\eta_+},H)$ determines a unitary quantum system. The pure states of this system are represented by the rays in $\sH_{\eta_+}$, the observables are given by Hermitian operators acting in $\sH_{\eta_+}$, and the dynamics is governed by the time-dependent Schr\"odinger equation defined by $H$ in $\sH_{\eta_+}$.

Applications of the above method of constructing metric operators and the modified inner products for various toy models have been explored in the literature. A comprehensive list of references published prior to 2010 is given in the review article \cite{review}. Here we confine our attention to a very simple example that was originally considered in \cite{jpa-2006a}.

Consider the case that $v(x)=0$, i.e., $H$ is the second derivative operator acting in $L^2(\R)$. It is well-known that $H$ is Hermitian and has a nonnegative real continuous spectrum. We can easily check that it satisfies the $\etap$-pseudo-Hermiticity relation~(\ref{ph-1}) for
    \be
    \etap:=e^{-\kappa\cP}=\cosh(\kappa)I-\sinh(\kappa)\cP,
    \label{etap=1}
    \ee
where $\kappa$ is an arbitrary real number, and $\cP$ is the parity operator (\ref{P-T=}). Equation~(\ref{etap=1}) defines a genuine metric operator. Substituting it in (\ref{mod-inn}) yields the following expression for the corresponding modified inner product.
    \be
    \bra \phi,\psi\kt_{\eta_+}=\cosh(\kappa)\int_{-\infty}^\infty \phi(x)^*\psi(x)dx-
    \sinh(\kappa)\int_{-\infty}^\infty \phi(x)^*\psi(-x)dx.
    \label{q1}
    \ee
For the standard position operator $X$, that is given by $X\psi(x):=x\psi(x)$, we can use (\ref{q1}) to show that
    \[\bra\phi,X\psi\kt_{\eta_+}-\bra X\phi,\psi\kt_{\eta_+}=2\sinh(\kappa)\int_{-\infty}^\infty x\phi(x)^*\psi(-x)dx.\]
This quantity differs from zero for $\kappa\neq 0$ and $\psi(x)=x\phi(x)=x\,e^{-x^2}$. Therefore,
as an operator acting in $\sH_{\eta_+}$,  $X$ is not Hermitian, unless if $\kappa=0$. The same holds for the standard momentum operator, $P:=-i\frac{d}{dx}$.

Clearly the positive square root of the metric operator (\ref{etap=1}) has the form $\rho=e^{-\kappa\cP/2}$. In view of this relation, it is easy to show that the operators $X_{\eta_+}$ and $P_{\eta_+}$ defined by
	\begin{align*}
	&X_{\eta_+}:=\rho^{-1}X\,\rho=\eta_{_+}^{-1}X=
	[\cosh(\kappa)I+\sin(\kappa)\cP]X,\\
	&P_{\eta_+}:=\rho^{-1}P\,\rho=\eta_{_+}^{-1}P=
	[\cosh(\kappa)I+\sin(\kappa)\cP]P,
	\end{align*}
are Hermitian operators acting in $\sH_{\eta+}$. Because $[X_{\eta_+},P_{\eta_+}]=iI$, we can take $X_{\eta_+}$ and $P_{\eta_+}$ to represent the position and momentum observables of the system defined by $(\sH_{\eta+},H)$. For all real numbers $\kappa$, this is just a free particle moving on a straight line. But for different choices of $\kappa$, we have different operators representing the position and momentum of the particle. This has some peculiar consequences. For example, for $\kappa\neq 0$, the spatially localized states of the particle correspond to a linear combination of two Dirac delta-functions rather than a single delta-function \cite{jpa-2006a}!

When $\sH$ is an infinite-dimensional Hilbert space the above constructions are valid provided that we impose some additional technical conditions. Specifically, the complete biorthonormal eigensystem $\{(\psi_n,\phi_n)\}$ should be bounded \cite{review}. This is equivalent to the condition that $\{\psi_n\}$ and $\{\phi_n\}$ are Riesz bases of $\sH$, which means that they can be mapped to an orthonormal basis by a bounded invertible operator \cite{review,BIT}. The boundedness of $\{(\psi_n,\phi_n)\}$ implies that the metric operator $\etap$ must be a positive automorphism, i.e., a positive invertible operator that is defined everywhere in $\sH$ (which makes it bounded) and has a bounded inverse \cite{review}.

It turns out that if we define the Schr\"odinger operator for the potential $ix^3$ as a linear operator (with maximal domain) acting in $\sH:=L^2(\R)$, then we cannot satisfy (\ref{ph-2}) or (\ref{ph-1}) using a bounded positive-definite operator $\etap$ that is inversely bounded. Therefore, strictly speaking, an appropriate metric operator does not exist for this potential  \cite{seigl}. As we explain below this is a mathematical technicality that can be circumvented by paying due attention to the role of the linear operators representing physical observables in quantum mechanics.

Consider redefining the Hilbert space $\sH$ and the operator $H$ in such a way that the new Hilbert space $\sH'$ includes the eigenvectors of $H$ and the new operator $H'$, which acts in $\sH'$, shares both the spectrum and eigenvectors of $H$, \cite{ptrs-2013}. Because we can only prepare state vectors which are superpositions of the eigenvectors of the relevant observables, as far as $H$ is concerned both $\sH$ and $\sH'$ include all the prepareable state vectors, and $H$ and $H'$ are equivalent as representations of a quantum mechanical observable.  As shown in Ref.~\cite{ptrs-2013}, for a given metric operator $\etap$, which may violate the conditions of boundedness or inverse boundedness, it is possible to construct $\sH'$ and $H'$ in such a way that they have the above-mentioned properties and in addition $H'$ be a Hermitian operator. Therefore although one cannot use $(\sH,H)$ to define a unitary quantum system directly, one can construct $\sH'$ and $H'$ which contain the same physically relevant ingredients and use $(\sH',H')$ to define a unitary quantum system.

\section{Singularities of the Metric Operators}

Consider the Hilbert space $\sH$ obtained by endowing $\C^2$ with the Euclidean inner product. The elements of $\sH$ and the linear operators acting in it can be respectively represented by $2\times 1$ and $2\times 2$ matrices in the standard basis of $\C^2$. Using the same symbol for the matrix representations and the corresponding vectors and operators, we consider constructing the most general metric operator for
    \be
    H:=\begin{bmatrix}
    0 & 1\\
    x^2 & 0 \end{bmatrix},
    \label{2-dim}
    \ee
where $x\in\R$. It is easy to show that for this operator,
    \begin{align}
    &E_n=(-1)^n x, && \psi_n=\frac{N_n}{\sqrt 2}\begin{bmatrix}
    (-1)^n\\
    x\end{bmatrix}, && \phi_n=\frac{1}{\sqrt 2 N_n^*}\begin{bmatrix}
    (-1)^n\\
    x^{-1}\end{bmatrix},
    \label{last-eq}
    \end{align}
where $n=1,2$ and $N_n$ are arbitrary nonzero complex coefficients possibly depending on $x$.

Inserting the last of Eqs.~(\ref{last-eq}) in (\ref{spec-rep}) and introducing $a_\pm:=(|N_2|^{-2}\pm|N_1|^{-2})/2$, we find
    \begin{align*}
    &\etap=\begin{bmatrix}
    a_+ & a_- x^{-1}\\
    a_- x^{-1} & a_+ x^{-2} \end{bmatrix}.
    \end{align*}
This relation identifies $x=0$ with a singularity of all possible metric operators for $H$.  Note also that $H$ loses the property of being diagonalizable precisely for this value of $x$. This is an example of what is called an exceptional point \cite{H,jmp-2008} or a non-Hermitian degeneracy \cite{berry}.

The term `exceptional point' is introduced by Kato in his study of the effects of perturbations of a linear operator on its spectral properties \cite{kato}. The following is a widely used definition of this concept which differs slightly from Kato's.
    \begin{defn}
Let $V$ be a vector space, $m$ be a positive integer, $H(x):V\to V$ be a linear operator depending on $m$ real parameters $x_1,x_2,\cdots,x_m$. We identify these with local coordinates of a point $x$ of a parameter space  (a smooth manifold) $M$. Suppose that for each $x\in M$ the eigenvalues of $H(x)$ have finite geometric multiplicity and form a countable set of isolated points of $\C$ that we denote by $E_n(x)$. Here $n$ is a spectral label taking values in a discrete set $\cN$. Let $\mu_n(x)$ be the geometric multiplicity of $E_n(x)$, i.e., the dimension of the span of eigenvectors of $H(x)$ that are associated with the eigenvalue $E_n(x)$. A point $x_0$ of $M$ is called an exceptional point of $H(x)$ if there are $n\in\cN$, $\epsilon\in\R^+$, and a parameterized curve in $M$, i.e., a continuos function, $\gamma:(-\epsilon,\epsilon)\to M$, such that $\gamma(0)=x_0$ and for all $t\neq 0$, $\mu_n(\gamma(t))\neq \mu_n(x_0)$.
    \label{exc-pt}
    \end{defn}
For the case that $V$ is endowed with the structure of an inner-product space, we can speak of the adjoint of $H(x)$ and decide whether it is Hermitian. If for all $x\in M$, $H(x)$ is a Hermitian operator, the geometric multiplicity of the eigenvalues $E_n(x)$ do not undergo discontinuous changes and an exceptional point cannot exist. Therefore, non-Hermiticity is a necessary condition for the emergence of an exceptional point.

It turns out that exceptional points have a number of interesting physical realizations. See for example \cite{H-S,H,berry,jmp-2008,muller-rotter,moiseyev-book} and references therein. In particular, they lead to certain geometric phases which have been the subject of intensive theoretical \cite{H-S,H,berry,muller-rotter,mailybaev,jmp-2008} and experimental studies \cite{DGHHHRR,SHS,DHKMRS} since the early 1990's.

The two-dimensional model (\ref{2-dim}) can be easily generalized to higher dimensional matrix Hamiltonians $H(x)$, \cite{jmp-2008}. If we choose the eigenvectors of $H(x)$ in such a way that they are nonsingular functions of $x$, then exceptional points appear as the singularities of the eigenvectors of $H(x)^\dagger$ and consequently the corresponding metric operator (\ref{spec-rep}).

Definition~\ref{exc-pt} introduces exceptional points in terms of a condition on the eigenvalues of $H(x)$. If this operator acts in a Hilbert space, we can speak of its spectrum. This is a subset of $\C$ that in addition to the eigenvalues may contain numbers that are not eigenvalues of $H(x)$. The latter constitute two disjoint sets called the continuous and the residual spectra of $H(x)$ \cite{reed-simon}. Hermitian operators have an empty residual spectrum. The same is true for a large class of non-Hermitian operators.

A natural question that arises in the study of non-Hermitian operators with a real spectrum is how to generalize the notions of diagonalizability and the metric operator for operators whose spectrum includes a continuous part. The first step in this direction was taken in Ref.~\cite{jmp-2005}. It involved a direct extension of the approach developed for operators with a discrete spectrum to the  imaginary $\cP\cT$-symmetric barrier potential,
\begin{equation}
   v(x)=\left\{\begin{array}{ccc}
   -i\zeta &{\rm for} & -1\leq x\leq 0,\\
   i\zeta & {\rm for} & 0<x\leq 1,\\
   0 & {\rm for} & |x|\leq 1, \end{array}\right. \hspace{1cm} \zeta\in\R.
        \label{pt-barrier}
        \end{equation}
To the best of our knowledge, this provided the first example of a $\cP\cT$-symmetric potential which admitted an optical realization \cite{RDM}. The next step was to carry out the same analysis for the delta-function potential with a complex coupling constant \cite{jpa-2005b},
	\begin{equation}
	v(x)=\fz\,\delta(x), \hspace{1cm}\fz\in\C.
	\label{z-delta}
	\end{equation}
	
The treatment of (\ref{pt-barrier}) and (\ref{z-delta}) that was offered in \cite{jmp-2005} and \cite{jpa-2005b} is perturbative in nature. But there is an important difference; for imaginary values of $\fz$ regardless of how small $|\fz|$ is, the perturbative calculation of the metric operator for (\ref{z-delta}) is obstructed by the emergence of a singularity. In the remainder of this section, we provide a general description of this phenomenon and its relation to spectral singularities that was originally noticed in \cite{jpa-2005b} and explored more thoroughly for the double-delta-function potential in \cite{jpa-2009}:
	\begin{equation}
	v(x)=\fz_-\,\delta(x+a)+\fz_+\,\delta(x-a), \hspace{1cm}\fz_\pm\in\C,~a\in\R^+.
	\label{double-delta}
	\end{equation}

Let $v_z:\R\to\C$ be a scattering potential depending on complex parameters $z_1,z_2,\cdots,z_m$, that we collectively denote by $z$, i.e., $z:=(z_1,z_2,\cdots,z_m)$. Suppose that the  Schr\"odinger operator $H_z:=-\frac{d^2}{dx^2}+v_z(x)$ acts in $L^2(\R)$ and has a real and purely continuous spectrum given by $[0,\infty)$, i.e., its  point and residual spectra are empty. Then the nonzero elements of the spectrum of $H_z$ correspond to the numbers $k^2$ appearing on the right-hand side of the Schr\"ondinger equation~(\ref{sch-eq}). These are associated with a linearly independent pair of solutions of this equation that we denote by $\psi^{(z)}_{k,a}$ with $a=1,2$;
	\be
	H_z\psi^{(z)}_{k,a}=k^2\psi^{(z)}_{k,a}.
	\label{sch-eq-2}
	\ee	
Because $\psi^{(z)}_{k,a}$ do not belong to $L^2(\R)$, they are not eigenvectors of $H_z$. We refer to them as `generalized eigenfunctions' of $H_z$. Similarly, we can construct generalized eigenfunctions of $H_z^\dagger:=-\frac{d^2}{dx^2}+v_z(x)^*$ that we denote by $\phi^{(z)}_{k,a}$. These satisfy
	\be
	H_z^\dagger\phi^{(z)}_{k,a}=k^2\phi^{(z)}_{k,a}.
	\label{sch-eq-3}
	\ee
	
We can generalize the notion of `diagonalizability' for $H_z$, by demanding the existence of an eigensystem $\{(\psi^{(z)}_{k,a},\phi^{(z)}_{k,a})\}$ which satisfy the following biorthonormality and completeness relations \cite{jmp-2005}.
	\begin{align}
	&\bra \phi^{(z)}_{k,a}|\psi^{(z)}_{q,b}\kt=\delta_{ab}\delta(k-q), &&
	\sum_{a=1}^2\int_0^\infty dk~ |\psi^{(z)}_{k,a}\kt\bra\phi^{(z)}_{k,a}|=I.
	\label{biortho-3}
	\end{align}
Similarly we generalize the expression~(\ref{spec-rep}) for the metric operator:
	\be
	\etap=\sum_{a=1}^2\int_0^\infty dk~ |\phi^{(z)}_{k,a}\kt\bra\phi^{(z)}_{k,a}|.
	\label{etap=111}
	\ee	
	
Now, we demand that $v_z(x)^*=v_{z^*}(x)$. Then it is easy to see that
    	\be
    	H_z^\dagger\psi^{(z^*)}_{k,a}=
	\left[-\frac{d^2}{dx^2}+v_z(x)^*\right]\psi^{(z^*)}_{k,a}=
	H_{z^*}\psi^{(z^*)}_{k,a}=k^2\psi^{(z^*)}_{k,a}.
    	\label{condi-H}
    	\ee
Because for each $k\in\R^+$, the Schr\"odinger equation $H_z^\dagger\psi=k^2\psi$ has two linearly independent solutions, Eqs.~(\ref{sch-eq-3}) and (\ref{condi-H}) imply that $\phi^{(z)}_{k,a}$ are linear combinations of $\psi^{(z^*)}_{k,a}$, i.e., there are $J_{ab}(k)\in\C$ such that
	\be
	\phi^{(z)}_{k,a}=\sum_{b=1}^2 J^{(z)}_{ab}(k)\psi^{(z^*)}_{k,b}.
	\label{q101}
	\ee
It is also not difficult to show that $\bra\psi^{(z^*)}_{k,a}|\psi^{(z)}_{q,b}\kt$ is proportional to $\delta(k-q)$, i.e., there are $K^{(z)}_{ab}(k)\in\C$ such that
	\be
	\bra\psi^{(z^*)}_{k,a}|\psi^{(z)}_{q,b}\kt=K^{(z)}_{ab}(k)\delta(k-q).
	\label{q102}
	\ee
Inserting (\ref{q101}) in the first equation in (\ref{biortho-3}) and making use of (\ref{q102}), we find \cite{jpa-2009}
	\be
	\bJ(k)^{(z)*}\bK^{(z)}(k)=\bI,
	\label{q103}
	\ee
where $\bJ^{(z)}(k)$ and $\bK^{(z)}(k)$ are $2\times 2$ matrices having $J^{(z)}_{ab}(k)$ and $K^{(z)}_{ab}(k)$ as their entries, and $\bI$ is the $2\times 2$ identity matrix. Similarly, using (\ref{etap=111}), (\ref{q101}), and (\ref{q103}), we obtain
	\begin{align}
	&\etap=\sum_{a,b=1}^2 \int_0^k dk~\cE_{ab}^{(z)}(k)|\psi_{k,a}^{(z^*)}
	\kt\bra\psi_{k,b}^{(z^*)}|,
	\label{etap-202}
	\end{align}
where $\cE_{ab}^{(z)}(k)$ are entries of the matrix $[\bK^{(z)}(k)\bK^{(z)}(k)^\dagger]^{-1}$.

Equation~(\ref{q103}) implies that $\det(\bK^{(z)}(k))\neq 0$. But in general there is no reason why this relation should hold for all $k$ and $z$. Explicit calculations for the potentials (\ref{z-delta}) and (\ref{double-delta}) show that the values of $k^2$ for which $\det(\bK^{(z)}(k))=0$ are precisely the spectral singularities of the Schr\"odinger operator $H_z$, \cite{jpa-2009}. We are not aware of a proof of this statement for a general complex scattering potential. The proof for the double-delta-function potential that was given in  \cite{jpa-2009} revealed a useful connection between spectral singularities and the transfer matrix $\bM(k)$ of scattering theory \cite{sanchez}. It turned out that $\det(\bK^{(z)}(k))$ was proportional to the $M_{22}(k)$ entry of $\bM(k)$ with a nonzero proportionality factor. This provided the key observation that immediately led to the explanation of the physical meaning of a spectral singularity \cite{prl-2009}. We give a detailed discussion of these developments in the next section.

We close this section by noting that according to  (\ref{etap-202}), spectral singularities are also singularities of the metric operator. In this sense they are generalizations of the phenomenon of exceptional points to the linear operators that possess a nonempty continuous spectrum.

\section{Scattering Theory and Spectral Singularities}

Consider a possibly complex scattering potential $v(x)$ satisfying (\ref{condi-1}). The left- and right-incident scattering solutions of the Schr\"odinger equation, that we respectively denote by $\psi^l_k(x)$ and $\psi^r_k(x)$, satisfy the following asymptotic boundary conditions.
    \begin{align}
    &\psi^l_{k}(x)\to
    \left\{\begin{array}{ccc}
    \sA^l(k)\left[e^{ik x}+R^l(k) e^{-ik x}\right]&{\rm as}& x\to-\infty,\\[6pt]
    \sA^l(k) T^l(k) e^{ik x} &{\rm as}& x\to \infty,
    \end{array}\right.
    \label{psi-l}\\[6pt]
    &\psi^r_{k}(x)\to
    \left\{\begin{array}{ccc}
    \sA^r(k) T^r(k) e^{-ik x}&{\rm as}& x\to-\infty,\\[6pt]
    \sA^r(k) \left[e^{-ik x}+R^r(k) e^{ik x}\right]&{\rm as}& x\to \infty,
    \end{array}\right.
    \label{psi-r}
    \end{align}
where $\sA^{l/r}$, $R^{l/r}$, and $T^{l/r}$ are in general complex-valued functions. Because the Schr\"odinger equation (\ref{sch-eq}) is linear, the choice of $\sA^{l/r}$ does not affect the physically measurable quantities. This is not the case for $R^{l/r}$ and $T^{l/r}$, which are known as the left/right reflection and transmission amplitudes. Their modulus squared, $|R^{l/r}|^2$ and $|T^{l/r}|^2$, determine the left/right reflection and transmission coefficients that can be measured in experiments.\footnote{Notice that some authors use the symbols $R^{r/l}$ and $T^{r/l}$ for reflection and transmission coefficients.}

A well known consequence of the linearity of the Schr\"odinger equation (\ref{sch-eq}) is that $T^l=T^r$, \cite{ahmed-2001,prl-2009,SM-2014}.\footnote{This arises from Wronskian identities \cite{prl-2009,SM-2014} and has nothing to do with the reality of the potential as claimed in \cite{inverse-scattering}.} We therefore use $T$ for $T^{l/r}$. It is also easy to see that $\psi^{l/r}_k$ coincide with the Jost solutions $\psi_{k\pm}$ for $\sA^{l/r}(k)T(k)=1$;
    \begin{align}
    &\psi_{k+}(x)\to
    \left\{\begin{array}{ccc}
    T (k)^{-1}\left[e^{ik x}+R^l(k) e^{-ik x}\right]&{\rm as}& x\to-\infty,\\
    e^{ik x} &{\rm as}& x\to \infty,
    \end{array}\right.
    \label{psi-plus}\\
    &\psi_{k-}(x)\to
    \left\{\begin{array}{ccc}
    e^{-ik x}&{\rm as}& x\to-\infty,\\
    T (k)^{-1}\left[e^{-ik x}+R^r(k) e^{ik x}\right]&{\rm as}& x\to \infty.
    \end{array}\right.
    \label{psi-minus}
    \end{align}
The existence of the Jost solutions implies that $T(k)\neq 0$, i.e., perfectly absorbing potentials \cite{muga} do not exist.

The coefficients of the $e^{\pm ikx}$ that appear on the right-hand side of (\ref{psi-plus}) and (\ref{psi-minus}) turn out to coincide with the entries of a $2\times 2$ complex matrix known as the transfer matrix.

Because $v(x)\to 0$ as $x\pm\infty$, every solution $\psi(x)$ of the Schr\"odinger equation (\ref{sch-eq}) satisfies
    \be
    \psi(x)\to A_\pm(k)e^{ikx}+B_\pm(k) e^{-ikx}\hspace{.5cm} {\rm as}\hspace{.5cm}x\to\pm\infty,
    \label{asymp-gen}
    \ee
where $A_\pm(k)$ and $B_\pm(k)$ are complex coefficients. The transfer matrix $\bM(k)$ is defined by
the relation
    \[\begin{bmatrix} A_+(k)\\ B_+(k)\end{bmatrix}=
    \bM(k)\begin{bmatrix} A_-(k)\\ B_-(k)\end{bmatrix}.\]
In light of (\ref{psi-plus}), (\ref{psi-minus}), and (\ref{asymp-gen}), we can relate the entries $M_{ij}(k)$ of the transfer matrix $\bM(k)$ with the reflection and transmission amplitudes. This results in \cite{prl-2009}
    \begin{align}
    &M_{11}=T-\frac{R^l R^r}{T}, && M_{12}=\frac{R^r}{T}, &&
    M_{21}=-\frac{R^l}{T}, && M_{22}=\frac{1}{T},
    \label{M=RT}
    \end{align}
which, in particular, imply $\det\bM(k)=1$. Furthermore, we can use these relations to express (\ref{psi-plus}) and (\ref{psi-minus}) in the form
    \begin{align}
    &\psi_{k+}(x)\to
    \left\{\begin{array}{ccc}
    M_{22}(k)e^{ik x}-M_{21}(k) e^{-ik x}&{\rm as}& x\to-\infty,\\
    e^{ik x} &{\rm as}& x\to \infty,
    \end{array}\right.
    \label{psi-plus-2}\\
    &\psi_{k-}(x)\to
    \left\{\begin{array}{ccc}
    e^{-ik x}&{\rm as}& x\to-\infty,\\
    M_{22}(k)e^{-ik x}+M_{12}(k) e^{ik x}&{\rm as}& x\to \infty.
    \end{array}\right.
    \label{psi-minu2}
    \end{align}
The following characterization of spectral singularities is a direct consequence of these equations.
    \begin{thm}
    \label{thm4-1}
    Let $v:\R\to\C$ and $H$ be as in Definition~\ref{defn-ss}, $\bM(k)$ be the transfer matrix of $v$, $M_{ij}(k)$ be the entries of $\bM(k)$, and $k_\star$ be a positive real number. Then $k_\star^2$ is a spectral singularity of $H$ (or $v$) if and only if $M_{22}(k_\star)=0$.
    \label{thm-ss}
    \end{thm}
    \begin{proof}
    $k_\star^2$ is a spectral singularity of $H$ whenever $\psi_{k_\star -}$ and $\psi_{k_\star +}$ are linearly dependent. According to (\ref{psi-plus-2}) and (\ref{psi-minu2}) and the fact that these equations determine $\psi_{k\pm}$ uniquely, this happens if and only if $M_{22}(k_\star)=0$.
    \end{proof}
Combining the statement of Theorem~\ref{thm4-1} with Eqs.~(\ref{M=RT}) yields the physical meaning of spectral singularities, namely that spectral singularities are the real and positive values of the energy $k_\star^2$ at which reflection and transmission amplitudes diverge \cite{prl-2009}. The latter is a characteristic property of resonances, for they satisfy the outgoing boundary conditions \cite{seigert}. As seen from (\ref{psi-plus-2}) and (\ref{psi-minu2}), for the cases that $k_\star$ corresponds to a spectral singularity and $\psi_{k_\star\pm}$ become linearly dependent, they also satisfy the outgoing boundary conditions.

The main distinction between the wave function for a resonance and the Jost solutions $\psi_{k_\star\pm}$ at a spectral singularity $k_\star^2$ is that, unlike the latter, the former satisfies the Schr\"odinger equation for a non-real value of $k^2$. Because the imaginary part of $k^2$ for a resonance determines its width, we can identify spectral singularities with the energies of certain zero-width resonances. Note, however, that spectral singularities determine genuine non-decaying scattering states with real and positive energy \cite{prl-2009}. This distinguishes them from the bound states in the continuum \cite{WV}. Although the latter are also associated with zero-width resonances, their wave function is a square-integrable solutions of the Schr\"odinger equation. For a discussion of other differences between spectral singularities and bound states in the continuum, see \cite{ap-2013}.

The fact that the reflection and transmission amplitudes and consequently the reflection and transmission coefficients $|R^{l/r}(k)|^2$ and $|T(k)|^2$ diverge for a resonance does not conflict with the well-known unitarity condition
    \be
    |R^{l/r}(k)|^2+|T(k)|^2=1,
    \label{unitary}
    \ee
because the $k$-value for a resonance is not real. For a spectral singularity, $k$ is real and (\ref{unitary}) is violated. This provides a simple proof of the following result.
    \begin{thm}
    Real potentials cannot support a spectral singularity.
    \end{thm}

In the standard formulation of quantum mechanics, the Hamiltonian operator $H$ is required to be Hermitian and the potential functions $v$ are necessarily real-valued. Therefore, they do not display spectral singularities. The same applies to the pseudo-Hermitian representation of quantum mechanics \cite{review} where $H$ may not be Hermitian but Hermitizable. This is because the presence of a spectral singularities obstructs the existence of a metric operator that achieves the Hermitization process. However, complex scattering potentials have a number of applications in other areas of physics. The primary example is the optical potentials used in modeling optically active material. This is the arena in which the role and implications of spectral singularities have so far been studied. We devote the next section to a brief description of the optical realizations of spectral singularities.

\section{Spectral Singularities in Optics}

Consider an isotropic charge-free linear medium whose electromagnetic properties changes along one direction, that we take to be the $x$-axis in a Cartesian coordinate system. We can encode these properties in the definition of the refractive index of the medium $\bn(x)$ which is a generally complex quantity. Suppose that we are interested in the propagation of a linearly polarized time-harmonic electromagnetic wave in this medium. If we choose our $y$-axis along the polarization direction, we can express the electric field in the form $\vec E(\vec r,t)=e^{-i\omega t}\sE(\vec r)\hat e_y$, where $\vec r:=(x,y,z)$, $\omega$ is the angular frequency of the wave, $\sE(\vec r)$ is a solution of
    \be
    \left[\nabla^2+k^2\bn(x)^2\right]{\sE}(\vec r)=0,
    \label{Helm-eq}
    \ee
$\hat e_y$ is the unit vector along the positive $y$-axis, $k:=\omega/c$ is the wavenumber, and $c$ is the speed of light in vacuum \cite{Born-Wolf}.

Equation~(\ref{Helm-eq}) admits solutions depending only on $x$; ${\sE}(\vec r)=\psi(x)$. In view of (\ref{Helm-eq}), $\psi(x)$ satisfies the Schr\"odinger equation~(\ref{sch-eq}) corresponding to the potential
    \be
    v(x):=k^2[1-\bn(x)^2].
    \label{opt-potential}
    \ee
If the medium is confined to a compact region in empty space, $\bn(x)=1$ for sufficiently large values of $|x|$. This together with the fact that $\bn$ is a complex-valued function imply that $v(x)$ is a (finite-range) complex scattering potential. Therefore, optical potentials (\ref{opt-potential}) provide a fertile ground for the investigation of the physical implications of spectral singularities. Ref.~\cite{prl-2009}, which offers the first such investigation, explores spectral singularities in a medium described by an optical potential of the form
(\ref{pt-barrier}).

The physical meaning of these spectral singularities is more easily understood for a simpler model that consists of a homogeneous optically active infinite planar slab of length $L$ placed in vacuum
\cite{pra-2009,pra-2011}. This corresponds to a complex barrier potential,
     \be
    v(x)=\left\{\begin{array}{ccc}
    \fz & {\rm for} & |x|\leq L/2,\\
    0 & {\rm for} & |x|> L/2,\end{array}\right.\hspace{1cm} \fz:=k^2(1-\fn^2),
    \label{barrier-potential}
    \ee
where $\fn$ stands for the refractive index of the slab.

Inside the slab, where $|x|\leq L/2$, the Schr\"odinger equation (\ref{sch-eq})  admits a solution of the form $\psi(x)=\sE_0 e^{ik\fn (x+L/2)}$, where $\sE_0$ is a constant. This corresponds to a right-going plane wave
    \[\vec E(\vec r,t)=\sE_0 e^{i[k\fn(x+L/2)-\omega t]}\hat e_y,\hspace{.5cm} |x|\leq L/2.\]
If we use $\eta$ and $\kappa$ to respectively denote the real and imaginary parts of $\fn$, so that $\fn=\eta+i\kappa$, we find
    \[|\vec E(\vec r,t)|^2=|\sE_0|^2 e^{-k\kappa (2x+L)},\hspace{.5cm} |x|\leq L/2.\]
In particular, as the wave travels through the slab, its intensity changes from $|\sE_0|^2$ to  $|\sE_0|^2 e^{-2k\kappa L}$, i.e., it undergoes an exponential loss or gain of intensity by a factor of $e^{-2k\kappa L}$ depending on whether $\kappa>0$ or $\kappa<0$. Because of this a medium that has a positive (respectively negative) value for $\kappa$ is called a lossy (respectively gain) medium. The factor $2k|\kappa|$ that determines the amount of the exponential loss (gain) per unit distance traversed by the wave is called the attenuation (respectively gain) coefficient. In terms of the wavelength, $\lambda:=2\pi/k$, this quantity takes the form $4\pi|\kappa|/\lambda$. In particular, the {\em gain coefficient} is given by \cite{silfvast}
	\be
	g:=-\frac{4\pi\kappa}{\lambda}.
	\label{gain-coef}
	\ee
	
Because the complex barrier potential (\ref{barrier-potential}) is exactly solvable, we can easily determine its transfer matrix and explore its spectral singularities.  This is done in Refs.~\cite{pra-2009,pra-2011}. Here we suffice to state that the relation $M_{22}(k_\star)=0$, which determines the spectral singularities $k_\star^2$ whenever $k_\star\in\R^+$, reduces to the following complex transcendental equation \cite{pra-2011}.
	\be
	e^{-2i\fn k_\star L}=\left(\frac{\fn-1}{\fn+1}\right)^2.
	\label{q001}
	\ee
The right-hand side of this relation is a well-known quantity in optics called the {\em reflectivity} $\cR$. If we compute the modulus (absolute-value) of both sides of (\ref{q001}) and use (\ref{gain-coef}) in the resulting expression, we obtain \cite{pra-2011}
	\be
	g=\frac{1}{2L}\ln \frac{1}{|\cR|^2}.
	\label{q002}
	\ee
This equation that is a consequence of the existence of a spectral singularity is one of the basic relations of laser physics known as the laser threshold condition \cite{silfvast}. The right-hand side of (\ref{q002}) is the minimum amount of gain necessary for a (mirrorless) slab laser to begin emitting laser light. It is called the threshold gain coefficient.

Every laser amplifies the background noise to sizable intensities and emits it as coherent electromagnetic radiation. This is precisely what a spectral singularity does, because it leads to infinite reflection and transmission coefficients that are capable of amplifying extremely week background electromagnetic waves to considerable intensities. The fact that the waves emitted from both sides of a slab laser have the same intensity and phase (are coherent) also follows from (\ref{q001}). This is indeed a general property of spectral singularities, because they are invariant under the space reflection (parity) $\cP$. Under $\cP$ the transfer matrix $\bM(k)$ of every scattering potential transforms as
    \be
    \bM(k) \stackrel{\cP}{\longleftrightarrow} \mathbf{\sigma}_1\bM(k)^{-1}\mathbf{\sigma}_1,
    \label{P-trans-M}
    \ee
where ${\sigma}_1 $ is the first Pauli matrix, i.e., the $2\times 2$ matrix with zero diagonal and unit off-diagonal entries, \cite{longhi-2010,jpa-2012}. According to (\ref{P-trans-M}), $M_{22}(k)$ is $\cP$-invariant. Therefore, the same holds for the spectral singularities that are given by the real and positive zeros of $M_{22}(k)$.

In Ref.~\cite{prl-2013}, we develop a nonlinear generalization of spectral singularities that apply to nonlinearities that are confined in space (have compact support.) It turns out that the mathematical relation describing these nonlinear spectral singularities for the above simple slab model supplemented with a weak Kerr nonlinearity yields an equation relating the output intensity $I$ of the slab laser to its gain coefficient \cite{pra-2013c}. For a typical optical gain medium \cite{silfvast}, which satisfies $|\kappa|\ll 1<\eta$, this equation takes the following form.
	\be
	I=\frac{f(\eta)(g-g_{th})}{\sigma\,g_{th}},
	\label{q003}
	\ee
where $f$ is a real-valued function taking strictly positive values, $g$ is the gain coefficient (\ref{gain-coef}),  $g_{th}$ is the threshold gain coefficient that is given by the right-hand side of (\ref{q002}), and $\sigma$ is the Kerr coefficient which, for generic gain media, takes small but positive values.

Because $f(\eta)>0$, $\sigma>0$, and $I\geq 0$, Eq.~(\ref{q003}) implies that there is no power emitted from a slab laser unless we have $g>g_{th}$, and for $g>g_{th}$ the intensity of emitted wave increases linearly as a function of $g-g_{th}$. Both of these statements are among the basic results of the physics of lasers. Here they follow as logical consequences of the purely mathematical condition of the existence of a nonlinear spectral singularity. Let us also mention that (\ref{q003}) has a more general domain of validity. In Ref.~\cite{SAP-2014}, we explore the consequences of the emergence of nonlinear spectral singularities for a weakly nonlinear $\cP\cT$-symmetric bilayer slab. This consists of a pair of adjacent infinite homogeneous planar slabs with complex-conjugate refractive index, $\eta\pm i\kappa$, so that one's gain (loss) is balanced by the other's loss (gain)\cite{jpa-2012,pra-2003a}. The laser output intensity computed using the condition of the appearance of a nonlinear spectral singularity is also given by (\ref{q003}), albeit with a different choice for the function $f$, \cite{SAP-2014}.
	
Another interesting development having its root in optical spectral singularities is the discovery of perfect coherent absorbers (CPA) which are also called antilasers \cite{antilaser1,antilaser2,longhi-2010,longhi,antilaser3}. These are optical devices that function as time-reversed lasers, i.e., they completely absorb coherent electromagnetic waves.

Under the time-reversal transformation (\ref{P-T=}), scattering potentials $v(x)$ and their transfer matrix $\bM(k)$ transform according to
    \begin{align}
    & v(x)\stackrel{\cT}{\longleftrightarrow}v(x)^*, &&
    \bM(k) \stackrel{\cT}{\longleftrightarrow} \mathbf{\sigma}_1\bM(k)^{*}\mathbf{\sigma}_1.
    \label{M-T-trans}
    \end{align}
In light of these relations, the time-reversal transformation $\cT$ converts an optical potential (\ref{opt-potential}) describing a gain media into that of a lossy medium, and induces the transformation:
    \[M_{11}(k)\stackrel{\cT}{\longleftrightarrow}M_{22}(k)^*.\]
This, in particular, means that the spectral singularities of $v(x)$ correspond to the real values of the wavenumber $k$ at which the $M_{11}(k)   $ entry of the transfer matrix of the time-reversed potential, $v(x)^*$, vanishes. At this wavenumber the optical system modeled by $v(x)^*$ serves as a CPA. In other words, CPA action is a realization of the spectral singularities of the time-reversed (complex-conjugate) optical potential \cite{longhi-2010,jpa-2012}.

\section{Concluding Remarks}

Spectral singularities were introduced by Naimark about sixty years ago and have since become a subject of research in operator theory. Given their interesting mathematical implications, it is quite surprising that their relevance to scattering theory and their physical meaning could not be understood earlier than in 2009. It turns out that the optics of gain media offers various physical models in which this concept can be realized. The study of the optical realization of spectral singularities shows that they form a mathematical basis for lasers. This observation could be made much earlier, had the optical physicists knew about spectral singularities. Indeed, the solution of the wave equations with outgoing boundary conditions, which leads to spectral singularities for real wavenumbers, has been employed in laser theory previously \cite{TSC}.

The discovery of the physical aspects of spectral singularities has boosted interest in their study particularly among physicists. During the past five years there have appeared a number of research publications on the subject. The following is a list of those that we did not elude to above.
    \begin{itemize}
    \item[-] Refs.~\cite{sokolov,samsonov,correa,chaos-cador} address some of the formal and conceptual aspects of the subject.
    \item[-] Refs.~\cite{jpa-2009,ahmed1,jpa-2011,jpa-2012,sinha} explore specific toy models supporting spectral singularities.
    \item[-] Refs.~\cite{pra-2011b,pra-2012} study the application of semiclassical approximation and perturbation theory for determining spectral singularities of non-homo\ geneous gain media with planar symmetry.
    \item[-] Refs.~\cite{pla-2010-prsa-2012,pra-2014b-pra-2014d} examine the optical spectral singularities in spherical and cylindrical geometries. In particular, \cite{pra-2014b-pra-2014d} offers a detailed and careful treatment of spectral singularities in the whispering gallery modes. These correspond to the cylindrical and spherical lasers.
    \item[-] Refs.~\cite{longhi-2009,li} discuss some of the applications of spectral singularities in condensed matter physics.
    \item[-] Refs.~\cite{liu,reddy} consider spectral singularities in certain optically active waveguides and elaborate on their regularization due to the presence of nonlinearities.
    \item[-] Ref.\cite{aalipour} offers an extension of the analysis of \cite{pra-2011} to waves with a non-normal incidence angle.
    \item[-] Refs.~\cite{ap-2014,pra-2014b,Garsia-calderon-2014} are some other publications that discuss spectral singularities.
    \end{itemize}

The recent development of a nonlinear generalization of spectral singularities \cite{prl-2013} has opened the way towards applications of this concept in the vast territory of nonlinear waves. The fact that the simple applications in effectively one-dimensional optical systems yield a mathematical derivation of the known behavior of the laser output intensity provides ample motivation for further study of nonlinear spectral singularities in other areas of physics.

\subsection*{Acknowledgment} I would like to express my gratitude to the organizers of the XXXIII Workshop on Geometric Methods in Physics, in particular Piotr Kielanowski, for their hospitality during this meeting. I am also indebted to Hamed Ghaemi-\ dizicheh for helping me locate and correct the typos in the first draft of the manuscript. This work has been supported by  the Scientific and Technological Research Council of Turkey (T\"UB\.{I}TAK) in the framework of the project no: 112T951, and by the Turkish Academy of Sciences (T\"UBA).

\end{document}